\documentclass{llncs}

\usepackage{enumerate}
\usepackage{url}
\usepackage{textcomp}
\usepackage{ifthen}
\usepackage{xargs}
\usepackage{amssymb}
\usepackage{textcomp}
\usepackage{amsmath}
\usepackage{comment}
\usepackage{graphicx}
\usepackage{color}

\newcommand{\ket}[1]{|#1\rangle}

\renewcommand{\P}[1] {
    \mathbf{P}\text{-}{#1}
}
\newcommand{\SUPERPOLY}[1] {
    \mathbf{SUPERPOLY}\text{-}{#1}
}
\newcommand{\SUBEXP}[2] {
    \mathbf{SUBEXP}^{#1}\text{-}{#2}
}

\newcommand{\BPP}[2] {
    \mathbf{BPP}_{#1}\text{-}{#2}
}
\newcommand{\BSUPERPOLY}[2] {
    \mathbf{BSUPERPOLY}_{#1}\text{-}{#2}
}
\newcommand{\BSUBEXP}[3] {
    \mathbf{BSUBEXP}^{#1}_{#2}\text{-}{#3}
}

\newcommandx{\OBDD}[2][1=1, 2]{
	\ifthenelse{\equal{#2}{}}{ }{#2\text{-}}\ifthenelse{\equal{#1}{1}}{
    }{#1\text{-}}\mathrm{OBDD}
}

\newcommandx{\NOBDD}[2][1=1, 2]{
	\ifthenelse{\equal{#2}{}}{ }{#2\text{-}}\ifthenelse{\equal{#1}{1}}{
    }{#1\text{-}}\mathrm{NOBDD}
}

\newcommandx{\POBDD}[2][1=1, 2]{
	\ifthenelse{\equal{#2}{}}{ }{#2\text{-}}\ifthenelse{\equal{#1}{1}}{
    }{#1\text{-}}\mathrm{POBDD}
}

\newcommandx{\QOBDD}[2][1=1, 2]{
	\ifthenelse{\equal{#2}{}}{ }{#2\text{-}}\ifthenelse{\equal{#1}{1}}{
    }{#1\text{-}}\mathrm{QOBDD}
}

\newcommandx{\BP}[1][1=1]{
	\ifthenelse{\equal{#1}{1}}{}{#1\text{-}}\mathrm{BP}
}

\newcommand{\id}{\mathrm{id}}

\newcommandx{\C}[3][1, 3]{
    C_{#3}^{#1}({#2})
}

\newcommand{\set}[2][ ]{\left\{#2 \ifthenelse{\equal{#1}{ }}{ }{~|~#1}\right\}}
\newcommand{\bool}{\set{0, 1}}

\newcommand{\bin}[1]{\mathrm{bin}\left(#1\right)}
\newcommand{\reorder}[1]{\mathtt{reordering}_{#1}}
\newcommand{\xorreorder}[1]{\mathtt{xor}\text{-}\mathtt{reordering}_{#1}}

\newcommand{\ceil}[1]{\lceil{#1}\rceil}
\newcommand{\floor}[1]{\left\lfloor{#1}\right\rfloor}

\newcommand{\Z}{\mathbb{Z}}
\newcommand{\N}{\mathbb{N}}

\newcommand{\BPOBDD}{\mathbf{BPOBDD}}
\newcommand{\BQOBDD}{\mathbf{BQOBDD}}
\newcommand{\EQ}{\mathrm{EQ}}
\newcommand{\SEQ}{\mathrm{SEQ}}
\newcommand{\REQ}{\mathrm{REQ}}
\newcommand{\MOD}{\mathrm{MOD}}
\newcommand{\MWS}{\mathrm{MWS}}
\newcommand{\WS}{\mathrm{WS}}
\newcommand{\MSW}{\mathrm{MSW}}
\newcommand{\PERM}{\mathrm{PERM}}
\newcommand{\PJ}{\mathrm{PJ}}
\newcommand{\RPJ}{\mathrm{RPJ}}

\begin{document}

    \markboth{K. Khadiev, A. Khadieva, A. Knop}
    {Exponential Separation between Quantum and Classical OBDDs, Reordering Method and Hierarchies}

    %
    %

    \title{Exponential Separation between Quantum and Classical Ordered Binary
            Decision Diagrams, Reordering Method and Hierarchies}

    \author{Kamil~Khadiev\\
        Institute of Computational Mathematics and IT,
        Kazan Federal University \\Kremlevskaya str, 35, Kazan, 420008, Russia\\
        \texttt{kamilhadi@gmail.com}\\
%
    Aliya~Khadieva\\
        Institute of Computational Mathematics and IT,
        Kazan Federal University \\Kremlevskaya str, 35, Kazan, 420008, Russia\\
        Faculty of Computing, University of Latvia,\\
        Raina bulvaris 19, Riga, LV-1586, Latvia\\
        \texttt{aliya.khadi@gmail.com}
%
 \\Alexander~Knop\\
        Department of Mathematics, University of California, San Diego\\
        9500 Gilman Dr, La Jolla, CA 92093-0112, USA\\
        \texttt{aaknop@gmail.com}
    }

    \maketitle


    \begin{abstract}
        In this paper, we study quantum Ordered Binary Decision Diagrams($\OBDD$) model; it is a restricted version of read-once quantum branching programs, with respect to ``width'' complexity. It
is known that the maximal gap between deterministic and quantum complexities is
exponential. But there are few examples of functions with such a gap. We present a new technique (``reordering'') for proving lower bounds and upper bounds for OBDD with an arbitrary order of input variables if we have similar bounds for the natural order. Using this transformation, we construct a total function $\REQ$
such that the deterministic $\OBDD$ complexity of it is at least $2^{\Omega(n /
\log n)}$, and the quantum $\OBDD$ complexity of it is at most $O(n^2/\log n)$. It is the biggest
known gap for explicit functions not representable by $\OBDD$s of a linear
width. Another function(shifted equality function) allows us to obtain a gap $2^{\Omega(n)}$ vs $O(n^2)$.

Moreover, we prove the bounded error quantum and probabilistic $\OBDD$ width hierarchies for
complexity classes of Boolean functions. Additionally, using  ``reordering'' method we extend a hierarchy for read-$k$-times Ordered Binary Decision Diagrams ($\OBDD[k]$) of polynomial width, for $k = o(n /
\log^3 n)$. We prove a similar hierarchy for bounded error probabilistic
$\OBDD[k]$s of polynomial, superpolynomial and subexponential width.

The extended abstract of this work was presented on International
            Computer Science Symposium in Russia, CSR 2017, Kazan, Russia, June
            8 -- 12, 2017 \cite{kk2017}
    \end{abstract}

    \textit{quantum computing, quantum OBDD, OBDD, Branching programs,
    quantum vs classical, quantum models, hierarchy, computational complexity,
    probabilistic OBDD}

    \section{Introduction}
    Branching programs are a well-known computation model for discrete functions. 
This model has been shown useful in a variety of domains, such as hardware verification, model checking, and other CAD applications~\cite{Weg00}. 

One of the most important types of branching programs is oblivious read once branching programs, also known as Ordered Binary Decision Diagrams, or $\OBDD$~\cite{Weg00}. This model is suitable for studying of data streaming algorithms that are actively used in industry.

One of the most useful measures of complexity of $\OBDD$s is ``width''. This measure is an analog of a number of states for finite automaton and $\OBDD$s can be seen as nonuniform finite automata (see for example~\cite{AG05}). As for many other computation models, it is possible to consider quantum $\OBDD$s, and during the last decade they have been studied vividly~\cite{aazksw2019part1,AGK01,nhk00,ss2005,s06,aakk2018,ikpy2017,ikpy2021,gy2017,gy2015,gy2018,kkm2017,kkkrym2017}.

In 2005 Ablayev, Gainutdinova, Karpinski, Moore, and Pollett~\cite{AGKMP05} have proven that for any total Boolean function $f$ the gap between the width of the minimal quantum $\OBDD$ representing $f$ and the width of the minimal deterministic $\OBDD$ representing $f$ is at most exponential. However, this is not true for partial functions~\cite{AGKY14,g15,agky16}. They have also shown that this bound could be reached for $\MOD_{p, n}$ function, that takes the value $1$ on an input iff number of $1$s modulo $p$ in this input is equal to $0$. Authors have presented a quantum $\OBDD$ of width $O(\log p)$ for $\MOD_{p, n}$ (another quantum $\OBDD$ of the same width has been presented in~\cite{AV08}). Additionally, they have proven that any deterministic $\OBDD$ representing $\MOD_{p, n}$ has the width at least $p$. However, a lower bound for a width of a deterministic $\OBDD$ that represents $\MOD_{p, n}$ is tight, and it was unknown if it is possible to construct a function with an exponential gap but an exponential lower bound for the size of a deterministic $\OBDD$ representing this function. It was shown
that Boolean function $\PERM_n$ did not have a deterministic $\OBDD$ representation of width less than $2^{\sqrt{n} / 2}/(\sqrt{n} / 2)^{3 / 2}$~\cite{kmw91}. In 2005 Sauerhoff and
Sieling~\cite{ss2005} presented a quantum $\OBDD$ of width $O(n^2\log n)$ representing $\PERM_n$ and three years later Ablayev, Khasianov, and Vasiliev~\cite{akv2008} improved this lower bound and presented a quantum $\OBDD$ for this function of width $O(n \log n)$. But as in the previous case, this separation does not give us a truly exponential lower bound for deterministic $\OBDD$s.

Nevertheless, if we fix an order of variables in the $\OBDD$, it is possible to prove the desired statement. For example, it is known that \textit{equality function}, or $\EQ$, does not have an $\OBDD$ representation of the size less than $2^n$ for some order and it has a quantum $\OBDD$ of width $O(n)$ for any order~\cite{akv2008}. Unfortunately, for some orders, the equality function has a small deterministic $\OBDD$s.

Proving lower bounds for different orders is one of the main difficulties of proving lower bounds on width of $\OBDD$s.  In the paper, we present a new technique that allows us to prove such lower bounds. Using the technique, we construct a Boolean function $g$ from a Boolean function $f$ such that if any deterministic $\OBDD$ representing $f$ with the natural order over the variables has width at least $d(n)$, then any deterministic $\OBDD$ representing $g$ has width at least $d(O(n/\log n))$ \textbf{for any order of input variables} and if there is a quantum $\OBDD$ of width $w(n)$ for $f$, then there is a quantum $\OBDD$ of width $O(w(\frac{n}{\log n}) \cdot \frac{n}{\log n})$ for the function $g$.
It means that if we have a function with some gap between quantum $\OBDD$ complexity and deterministic $\OBDD$ complexity for some order, then we can transform this function into a function with almost the same gap but for all the orders. We call this transformation ``reordering''. The idea which is used in
the construction of the transformation is similar to the idea of a
transformation from~\cite{k15,DBLP:journals/jsyml/Krajicek08}.

We prove five groups of results using the transformation. At first, we consider the result of the transformation applied to the equality function (we call the new function \textit{reordered equality} or $\REQ_q$). We prove that $\REQ_q$ does not have a deterministic $\OBDD$ representation of width less than $2^{\Omega(\frac{n}{\log n})}$ and there is a bounded error quantum $\OBDD$ of width $O(\frac{n^2}{\log^2 n})$, where $n$ is a length of an input. As a result, we get a more significant gap between a width of quantum $\OBDD$s and width of deterministic $\OBDD$s than this gap for the $\PERM_n$ function. We prove such a gap for all the orders in contrast with a gap for $\EQ_n$, and we prove a better lower bound for deterministic $\OBDD$s than the lower bound for the $\MOD_{p, n}$ function.

Additionally, we considered \textit{shifted equality function} ($\SEQ_n$). We prove that $\SEQ_n$ does not have a deterministic $\OBDD$ representation of the width less than $2^{\Omega(n)}$ and there is a bounded error quantum $\OBDD$ with width $O(n^2)$. Note that the lower bound for the width of the minimal $\OBDD$ representing $\SEQ_n$ is better than for $\REQ_q$ but the upper
bound for the width of the minimal quantum $\OBDD$ representation is much better.

Using properties of $\MOD_{p, n}$, $\REQ_q$, and \textit{mixed weighted sum function} ($\MWS$) introduced by~\cite{s2005}, we prove a width hierarchy for classes of Boolean functions computed by bounded error quantum $\OBDD$s. We prove three hierarchy theorems: 
\begin{enumerate}
    \item the first of them and the tightest works for width up to $\log n$;
    \item the second of them is slightly worse than the previous one, but it works for width up to $n$;
    \item and finally the third one with the widest gap works for width up to $2^{O(n)}$.
\end{enumerate}
Similar hierarchy theorems are already known for deterministic $\OBDD$s \cite{AGKY14,AK13},  nondeterministic $\OBDD$s \cite{agky16}, and  $\OBDD[k]$s \cite{k15,ki2017,aakk2018}. Additionally, we present similar hierarchy theorems for bounded error probabilistic $\OBDD$s in the paper.

Hierarchies for quantum-model complexity classes and gaps for deterministic and quantum complexities were shown by researchers for automata models \cite{sy2014,kiy2018,qy2009} and other streaming (automata-like) models \cite{l2006,l2009,kkm2018,kk2019,kk2019disj,kk2022,kkzmkry2022}.

The fourth group of results is an extension of hierarchies by a number of tests for deterministic and bounded error probabilistic $\OBDD[k]$s of polynomial size. There are two known results of this type:
\begin{itemize}
    \item The first is a hierarchy theorem for $\OBDD[k]$s that was proven by Bollig, Sauerhoff, Sieling, and Wegener \cite{bssw96}. They have shown that $\P{\OBDD[(k - 1)]} \subsetneq         \P{\OBDD[k]}$ for $k = o(\sqrt{n} \log^{3 / 2} n)$;
    \item The second one was proven in \cite{Kha16} it states that         $\P{\OBDD[k]} \subsetneq \P{\OBDD[(k \cdot r)]}$ for $k = o(n / \log^2 n)$ and $r = \omega(\log n)$.
\end{itemize}
We partially improve both of these results, proving that $\P{\OBDD[k]} \subsetneq \P{\OBDD[2k]}$ for $k = o(n / \log^3 n)$. Our result improves the first one because it holds for bigger $k$, and the second one, because of a smaller gap between classes. The proof of our hierarchy theorem is based on properties of the Boolean function called \textit{reordered pointer jumping}, which is ``reordering'' of \textit{pointer jumping} function defined in~\cite{NW93,bssw96}.

Additionally, we partially improve a similar result of Hromkovich and Sauerhoff~\cite{hs2003} for a more general model, for probabilistic oblivious $\BP[k]$. They have proven such a hierarchy for $k \leq \log \frac{n}{3}$. We show similar hierarchy for polynomial size bounded error probabilistic $\OBDD[k]$s with error at most $1 / 3$ for $k = o(n^{1/3} / \log n)$. 

\subsection*{Structure of the paper} 
Section \ref{sec:prlmrs} contains descriptions of models, classes, and other necessary definitions. Discussion about the reordering method and applications for quantum $\OBDD$s is located in Section~\ref{sec:deord}. Section~\ref{sec:com-obdd} contains an analysis of properties of a function that guarantee the existence of a small commutative $\OBDD$ representation of this function. In Section~\ref{sec:gap} we explore the gap between quantum and
deterministic $\OBDD$ complexities. The width hierarchies for quantum and probabilistic $\OBDD$s are proved in Section~\ref{sec:hierarchies}. Finally, Section~\ref{sec:kobdd-hrch} contains applications of the reordering method and hierarchy results for deterministic and probabilistic $\OBDD[k]$s. Section \ref{sec:conc} concludes the paper.

    \section{Preliminaries}\label{sec:prlmrs}
    Ordered binary decision diagrams, or $\OBDD$s, is a well-known way to represent Boolean functions. This model is a restricted version of Branching Program \cite{Weg00}. A branching program over a set $X = \set{x_1, \dots, x_n}$ of $n$ Boolean variables is a directed acyclic graph $P$ with one source node $s$. Each inner node $v$ of $P$ is labeled by a variable $x_i \in X$, each edge of
$P$ is labeled by a Boolean value. For each node $v$ labeled by a variable $x_i$, $v$ has outgoing edges labeled by $0$ or $1$, and each sink of this graph is labeled by a Boolean value. A branching program $P$ called \textit{deterministic} iff for each inner node there are exactly two outgoing edges labeled by $0$ and $1$, respectively.

We say that a branching program $P$ accepts $\sigma \in \bool^n$ iff there a exists a path, called \textit{accepting} path, from the source to a sink labeled by $1$, such that in the all nodes labeled by a variable $x_i$ this path goes along an edge labeled by $\sigma(i)$. A branching program $P$ represents a Boolean function $f : \bool^n \to \bool$ if for each $\sigma \in \bool^n$ $f(\sigma) = 1$ holds iff $P$ accepts $\sigma$. The \textit{size} of a branching program $P$ is a number of nodes in the graph.

A branching program is \textit{leveled} if the nodes can be partitioned into levels $V_1$, \dots, $V_\ell$, and $V_{\ell + 1}$ such that all the sinks belong to $V_{\ell + 1}$, $V_1 = \set{s}$, and nodes in each level $V_j$ with $j \le \ell$ have outgoing edges only to nodes in the next level $V_{j + 1}$. The \textit{width} $w(P)$ of a leveled branching program $P$ is the maximum of the number of nodes in levels of $P$, i.e. $w(P) = \max\limits_{1 \le j \le \ell + 1} |V_j|$. A leveled branching program is called \textit{oblivious} if all the inner nodes of each level are labeled by the same variable.

A branching program is called a \textit{read}-$k$ branching program if each variable is tested on each path only $k$ times. A deterministic oblivious leveled read once branching program is also called the ordered binary decision diagram. Note that $\OBDD$ reads variables on all the paths in the same order $\pi$.  For a fixed order $\pi$ we call an $\OBDD$ that reads in this order a $\OBDD[1][\pi]$. Let us also denote the natural order over the variables $\set{x_1, \dots, x_n}$ as $\id = (1, \dots, n)$. A branching program is called $\OBDD[k]$ if it is a read-$k$ oblivious branching program that consists of $k$ layers, such that each layer is a $\OBDD[1][\pi]$, possibly with many sources, for some order $\pi$. If we want to emphasize that a $\OBDD[k]$ has order $\pi$, then we write $\OBDD[k][\pi]$

Let $\mathrm{tr}_P : \set{1, \dots, n} \times \set{1, \dots, w(P)} \times \bool \to \set{1, \dots, w(P)}$ be a transition function of an $\OBDD$ $P$, where $n$ is length of input and $w(P)$ is width of $P$. We assume that all levels of $P$ contain $w(P)$ nodes. If some level contains less than $w(P)$ nodes, then we add additional dummy nodes to the level.  An $\OBDD$ $P$ is called \textit{commutative} iff for any order $\pi'$ we can construct an $\OBDD$ $P'$ by only reordering of the transition function and $P'$ still computes the same function. More formally, we call a $\OBDD[\pi]$ commutative iff for any order $\pi'$ a $\OBDD[\pi']$ $P'$, defined by a transition function $\mathrm{tr}_{P'}(i, a, b) = tr_{P}(\pi^{-1}(\pi'(i)), a, b)$, represents the same function as $P$. Additionally, we call a~$\OBDD[k]$ commutative if each layer is a commutative $\OBDD$.

Nondeterministic $\OBDD$ or $\NOBDD$ is a $\OBDD$ such that nodes can have more than one outgoing edges with the same label. A nondeterminstic $\OBDD$ accepts an input if there is at least one path from the source node to a sink node labeled by $1$. Now let us define probabilistic $\OBDD$ or $\POBDD$. $\POBDD$ over a set $X = \set{x_1, \dots, x_n}$ is a nondetermenistic $\OBDD$ with a special mode of acceptance. We say that $\POBDD$ is a bounded error representation of a Boolean function $f : \bool^n \to  \bool$ iff for
every $\sigma \in \bool^n$ the following recursion procedure returns $f(\sigma)$ with probability at least $\frac{1}{2}+\varepsilon$, for some $\varepsilon>0$:
\begin{enumerate}
    \item Initially it starts from the source of the $\POBDD$;
    \item If the current node is a sink, then it returns the value of its label;
    \item Let the current node be labeled by $x_i$ if there are no outgoing edges with label $\sigma(i)$ from the current node, then it returns $0$;
    \item Otherwise, it chooses randomly an edge labeled by $\sigma(i)$ from the current node to a node $u$, consider $u$ as the current node, and it goes to the step~$2$.
\end{enumerate}

I. Wegener's book \cite{Weg00} contains detailed definitions and more information on nondeterministic and probabilistic OBDDs.

Let us define quantum $\OBDD$s or $\QOBDD$s \cite{AGKMP05,AGK01}, Figure \ref{fig:qobdd}. You can find more detailed information about quantum computing in \cite{sy2014,ay2015}. For a given $n > 0 $, a $\QOBDD$ $P$ of a width $w$, is a $4$-tuple $P = (T, q_0 , \mathrm{Accept}, \pi)$, where 
\begin{itemize}
    \item $ T = \set{(G_i^0, G_i^1)}_{i = 1}^n$ is a sequence of pairs of (left) unitary matrices representing the transitions applying on the $i$-th step, where a choice of $G_i^0 $ or $G_i^1 $ is determined by the corresponding input bit;
    \item $\mathrm{Accept} \subseteq \set{1, \dots, w} $ is a set of accepting states;
    \item $\pi$ is a permutation of $ \set{1, \dots, n} $ defining the order over the input variables.
\end{itemize}

For any given input $ \sigma \in  \bool^n $, the computation of $P$ on $\sigma$ can be traced by a vector from $w$-dimensional Hilbert space over the field of complex numbers. The initial one is $\ket{\psi}_0 = \ket{q_0}$. On each step $j$, we test the input bit $x_{\pi(j)}$ and then the corresponding unitary operator is applied: $\ket{\psi}_j = G_j^{x_{\pi(j)}}\cdot (\ket{\psi}_{j - 1})$, where $\ket{\psi}_{j - 1}$ and $\ket{\psi}_j$ represent the states of the system after the $(j - 1)$-th and $j$-th steps, respectively. At the end of the computation, the program $P$ measures qubits. The accepting probability of $P$ on an input  $\sigma$ is $\sum\limits_{i \in \mathrm{Accept}} v^2_i$, where $(v_1, \dots, v_w) =  \ket{\psi}_n$. We say that a function $f : \bool^n \to \bool$ has a bounded error $\QOBDD$ representation iff for any $\sigma \in \bool^n$ and some $\varepsilon>0$ holds  
\begin{itemize}
    \item if $f(\sigma) = 1$, then the accepting probability of $P$ is at least $\frac{1}{2}+\varepsilon$ and 
    \item if $f(\sigma) = 0$, then the accepting probability of $P$ is at most $\frac{1}{2}-\varepsilon$.
\end{itemize}

\begin{figure}
\begin{center}
\includegraphics[width=10cm]{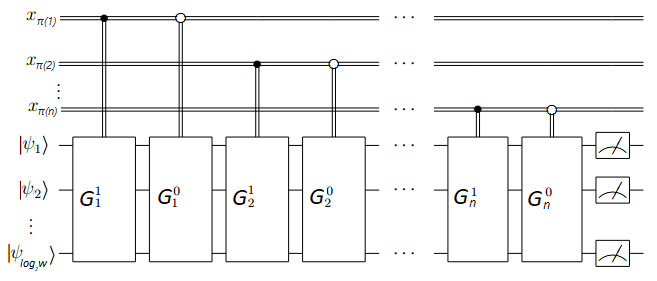}
\end{center}
\caption{Quantum OBDD.}
\label{fig:qobdd}
\end{figure}

Similarly to commutative deterministic $\OBDD$s we may define commutative $\QOBDD$s. $\QOBDD$ $P$ is called \textit{commutative} iff for any permutation $\pi'$ we can construct equivalent $\QOBDD$ $P'$ by only reordering matrices
$G$. Formally, it means that for any order $\pi'$, $P' = (T', q_0 , \mathrm{Accept}, \pi)$ is a bounded error representation of
the same function as $P$ where $T' = \set{\left(G_{\pi^{-1}(\pi'(i))}^0, G_{\pi^{-1}(\pi'(i))}^1\right)}^n_{i=1}$. We call a $\QOBDD[k]$ commutative if each layer of this program is a commutative $\QOBDD$. See Figure \ref{fig:comm-qobdd}.

\begin{figure}
\begin{center}
\includegraphics[width=13cm]{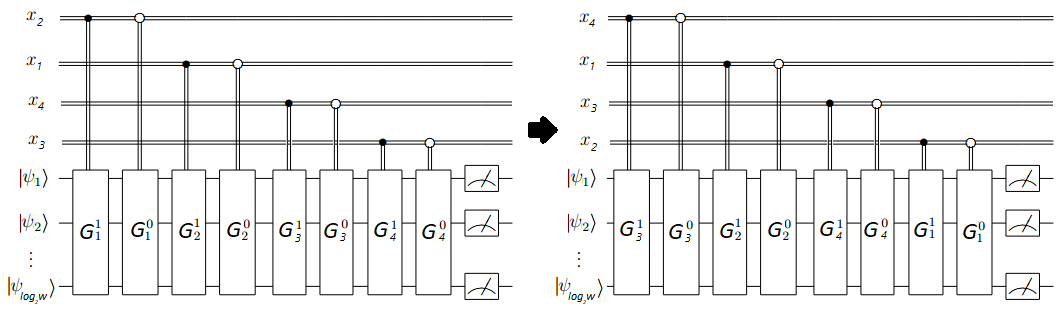}
\end{center}
\caption{Commutative QOBDD. The first QOBDD has order $(2,1,4,3)$. The second QOBDD has order $(4,1,3,2)$ and represents the same Boolean function.}
\label{fig:comm-qobdd}
\end{figure}

\subsection{Quantum Fingerprinting}\label{sec:def-fingerprint}

Let us present some basic concepts of quantum fingerprinting technique from \cite{af98,an2008,an2009,akv2008}. This technique is used in proofs from Sections \ref{sec:com-obdd} and ref{sec:gap}.

For the problem being solved we choose some cardinal $m$, an error probability bound $\varepsilon > 0$, fix $t = \lceil(2/\varepsilon) \ln 2m\rceil$, and construct a mapping $g : \{0, 1\}^n\to \mathbb{Z}$. Then for arbitrary binary string $\sigma = (\sigma_1 \dots \sigma_n)$ we create it's fingerprint $|h_\sigma\rangle$ composing $t$ single qubit fingerprints $|h_\sigma^i\rangle$:
\[|h_\sigma^i\rangle=cos\frac{2\pi k_i g(\sigma)}{m}|0\rangle + sin\frac{2\pi k_i g(\sigma)}{m}|0\rangle, \quad\quad\quad 
|h_\sigma\rangle=\frac{1}{\sqrt{t}}\sum_{i=1}^{t}|i\rangle|h^i_{\sigma}\rangle\]

Here the last qubit is rotated by $t$ different angles about the $\hat{y}$ axis of the Bloch sphere. The chosen parameters $k_i \in\{1\dots,m-1\}$, for $i\in\{1\dots t\}$ are ``good'' in the following sense. A set of parameters $K = \{k_1,\dots, k_t\}$ is called ``good'' for $g\neq 0 \mod m$ if
\[\frac{1}{t^2}\left(\sum_{i=1}^t cos\frac{2\pi k_i g}{m}\right)^2<\varepsilon\]
The left side of inequality is the squared amplitude of the basis state $|0\rangle^{\otimes \log_2 t} |0\rangle$ if the operator
$H^{\otimes \log_2 t}\otimes I $ has been applied to the fingerprint $|h_\sigma\rangle$. Informally, that kind of set guarantees, that
the probability of error will be bounded by a constant below $1$.

The following lemma from \cite{akv2008,an2008,an2009} proves the existence of a ``good'' set .
\begin{lemma}[\cite{akv2008}]
There is a set $K$ with $|K| = t = \lceil(2/\varepsilon) \ln 2m\rceil$  which is ``good'' for all $g\neq 0 \mod m$.
\end{lemma}

We use this result for fingerprinting technique \cite{akv2008} choosing the set $K = \{k_1,\dots, k_t\}$ that is ``good'' for all $g = g(\sigma)\neq 0$. It allows to distinguish those inputs whose image is $0$ modulo $m$ from the others.

That hints on how this technique may be applied:
\begin{enumerate}
\item We construct $g(x)$, that maps all acceptable inputs to $0$ modulo $m$ and others to arbitrary non-zero (modulo $m$) integers.

\item After the necessary manipulations with the fingerprint the $H^{\otimes \log_2 t}$ operator is applied to the first $\log_2 t$ qubits. This operation ``collects'' all of the cosine amplitudes at the all-zero state. That is, we obtain the state of type
\[|h'_\sigma\rangle=\frac{1}{t}\sum_{i=1}^{t}cos\left(\frac{2\pi k_i g(\sigma)}{m}\right) |00\dots 0\rangle|0\rangle + \sum_{i=2}^{2t}\alpha_i|i\rangle\]
\item This state is measured in the standard computational basis. Then we accept the input if the outcome is the all-zero state. This happens with probability
\[Pr_{accept}(\sigma)=\frac{1}{t^2} \left(\sum_{i=1}^{t}cos\frac{2\pi k_i g(\sigma)}{m}\right)^2,\]
which is $1$ for inputs, whose image is $0 \mod m$, and is bounded by $\varepsilon$ for the others.
\end{enumerate} 

   \section{Reordering Method}\label{sec:deord}
   As it was mentioned before, one of the biggest issues in proving lower bounds on 
the $\OBDD$ complexity of 
a function is proving these lower bounds for different orders. In this section, we suggest a method that 
is called ``reordering''. The method allows us to construct a transformation of a Boolean function 
$f : \bool^q \to \bool$ into a  partial function $\reorder{f} : \bool^n \to \bool$, such that 

\begin{itemize}
    \item $n = q \ceil{1 + \log q}$ \footnote{We use $\log$ to denote logarithms base
        $2$.};
    \item If any $\OBDD[\pi]$ representation of $f$ has width at least
        $d(q)$, then any $\OBDD$ representation of $\reorder{f}$ has width at
        least $d(q)$;
    \item If there is a bounded error commutative $\QOBDD$ representation of $f$
        of width $w(q)$, then there is a bounded error $\QOBDD$ representation
        of $\reorder{f}$ of width $w(q) \cdot q$.
\end{itemize}

We construct a function $\reorder{f}:\{0,1\}^{q\lceil \log q + 1\rceil}\to\{0,1\}$ by the following way:
\[
    \reorder{f}(z_{1, 1}, \dots, z_{1, l}, \dots, z_{q, 1}, \dots, z_{q, l},
        y_1, \dots, y_q) = f(y_{\theta^{-1}(1)},\dots,y_{\theta^{-1}(q)}),\]
where $l = \ceil{\log q}$, $\theta(i)=\bin{z_{i, 1}, \dots, z_{i, l}} + 1$ and $\bin{a_1, \dots, a_l}$ is a natural number with
binary representation $a_1 \dots a_l$. The function $\reorder{f}$ is defined on an input $(z_{1, 1}, \dots, z_{1, l}, \dots, z_{q, 1}, \dots, z_{q, l}, y_1, \dots, y_q)$ iff $\set{\theta(1), \dots, 
      \theta(l)} = \set{1, \dots, q}$. This condition means that $\theta^{-1}$ exists.

Similarly, we define 
\[
    \xorreorder{f}(z_{1, 1}, \dots, z_{1, l}, \dots, z_{q, 1}, \dots, z_{q, l},
                y_1, \dots, y_q) = f(y_{\theta^{-1}(1)},\dots,y_{\theta^{-1}(q)}),\]

where $\theta(i)=\bin{\bigoplus\limits_{j = 1}^i z_{j, 1}, 
                              \dots,
                              \bigoplus\limits_{j = 1}^i z_{j, l}} + 1$, for $i\in\{1,\dots,q\}$.

\begin{theorem} \label{th:d-obdd}
 Let $k$ be an integer, $\theta$ be a permutation of $\set{1, \dots, q}$.
 
  If $f : \bool^q \to \bool$ is a Boolean function such that any $\OBDD[k][\theta]$ representation of $f(x_1, \dots, x_q)$  has width at least $d$, then any $\OBDD[k]$ representation of $\reorder{f}(z_{1, 1}, \dots, z_{q, l}, y_1, \dots, y_q)$ ($\xorreorder{f}(z_{1, 1}, \dots, z_{q, l}, y_1, \dots, y_q)$) has width at least $d$.
\end{theorem}
\begin{proof}
    Proofs for $\reorder{f}$ and $\xorreorder{f}$ are almost the same. Here we present only the proof for $\reorder{f}$.

    Let us assume that there is a $\OBDD[k][\pi]$ $P$ for $\reorder{f}$ of width $d' < d$, where $\pi$ is permutation of $\{1,\dots,n\}$, $n=q(l+1)$. Let we meet $y_1,\dots,y_q$ variables in order $j_1,\dots,j_q$ when we consider variables $z_{1,1},\dots,z_{q,l},y_1,\dots,y_q$ in order $\pi$.
    
   Let us fix any order $\theta$ that is permutation of $\{1,\dots,q\}$. Then we consider only inputs $\Sigma\subset\{0,1\}^n$. Any input $\sigma\in\Sigma$ is such that $y_{j_1}, \dots ,y_{j_q}$ has addresses $\theta(1), \dots, \theta(q)$. Formally it means that  $\bin{z_{j_i, 1}, \dots, z_{j_i, l}} = \theta(i)$, for $i\in\{1,\dots,q\}$. 
    
    It is easy to see that if we consider $P'$ equal to $P$ such that we consider only inputs from $\Sigma$ and with all the variables $y_{j_i}$ replaced by $x_{\theta(i)}$. $P'$ is a $\OBDD[k][\theta]$ of width at most $d' < d$ that computes $f(x_1,\dots,x_q)$. This is a contradiction with the fact that any $\OBDD[k]$ that represents $f(x_1, \dots, x_q)$ has width at least $d$.
\end{proof}

\begin{theorem}\label{th:deord-classic}
    Let $f : \bool^q \to \bool$ be a Boolean function and $k$ be a positive integer. If there is a commutative $\OBDD[k]$ (bounded error commutative $\POBDD[k]$ or commutative $\NOBDD[k]$) for $f$ of width $d$, then there are $\OBDD[k]$ (bounded error $\POBDD[k]$ or $\NOBDD[k]$) representations of $\xorreorder{f}$ and $\reorder{f}$ of width $d \cdot q$.
\end{theorem}
\begin{proof} 
    Let $P$ be a commutative deterministic $\OBDD[k]$ of width $d$ representing a Boolean function $f$. We construct a deterministic $\OBDD[k]$s $P_1$ and $P_2$ of width $q \cdot d$ representing $\reorder{f}$ and $\xorreorder{f}$, respectively. $P_1$ and $P_2$ read variables in the following order: $z_{1, 1}, \dots, z_{1, l}, y_1,z_{2, 1}, \dots, z_{2, l}, y_2, \dots, z_{q, 1}, \dots, z_{q, l}, y_q$; both of them have $q \cdot d$ nodes on each level, each of them corresponds to a pair $(a, b)$, where $a \in \bool^l$ and $b \in \set{1,\dots, d}$, and both of them have $q$ stages. $a$ is an address of a following variable $y_i$, $b$ corresponds to some node of a level from the program $P$. Let us describe a computation on the stage $i$ on an input $\sigma\in\{0,1\}^{ql+q}$.
    \begin{description}
        \item[reordering:] At the beginning of the stage $P_1$ is in the node $(0, b)$ for some $b$. While reading $z_{i, 1}, \dots, z_{i, l}$ the program stores these variables in the first component of the current node. Formally, the node $(a',b)$ of a level $j$ of the stage has $1$-edge leading $(a'+2^j,b)$ node of the next level and $1$-edge leading $(a',b)$ node of the next level. After we have read all these bits, we reached  the node $(a, b)$. Then we apply a transition of $P$ for $(a+1)$-th variable to node of a level (state) $b$. Formally, if the transition function of $P$ is such that $b' = tr_P(\pi^{-1}(a + 1), b, y_i)$, then we go to the node $(0, b')$.
            
            Let $\sigma(z_{i', j'})$ be a bit of the input $\sigma$ that corresponds to a variable $z_{i',j'}$. In the case when all addresses $\bin{\sigma(z_{i, 1}), \dots, \sigma(z_{i, l})}$ are different numbers from $\set{1, \dots, q}$ the program $P_1$ just emulates the work of $P'$. Here $P'$ is $\OBDD[k][\pi]$ that is constructed from $P$ by permutation of the transition function of $P$ with respect to the order
            \[
                 \pi = \left(
                          \bin{\sigma(z_{1, 1}), 
                               \dots,
                               \sigma(z_{1, l})} + 1,
                      \right.\\
                          \dots,\\
                      \left.
                          \bin{\sigma(z_{q, 1}), 
                               \dots,
                               \sigma(z_{q, l})} + 1
                      \right).
            \]
            By the definition of the commutative $\OBDD[k]$, the program $P'$ computes the same function as $P$. Therefore, $P_1$ returns the same result. And by the definition of the function $\reorder{f}$, $P_1$ computes $\reorder{f}$.

         \item[xor-reordering:] 
$P_1$ is in the node $(0, b)$ for some $b$. While reading $z_{i, 1}, \dots, z_{i, l}$ the program stores xor of these variables with bits of the first component of the current node. Formally, the node $(a',b)$ of a level $j$ of the stage has two edges. The edge with label $z$ leading $(a'',b)$ node of the next level, for $a'=\bin{a'_1,\dots,a'_l},a''=\bin{a'_1,\dots,a'_{j-1},a'_{j}\oplus z, a'_{j+1},\dots,a'_{l}}$. After we have read all these bits, we reached  the node $(a, b)$. Then we apply a transition of $P$ for $(a+1)$-th variable to node of a level (state) $b$. Formally, if the transition function of $P$ is such that $b' = tr_P(\pi^{-1}(a + 1), b, y_i)$, then we go to the node $(a, b')$.

            In the case when all addresses $\bin{\bigoplus\limits_{i = 1}^1 \sigma(z_{i, 1}), 
                               \dots,
                               \bigoplus\limits_{i = 1}^1 \sigma(z_{i, l})} + 1$ are different numbers from $\set{1, \dots, q}$ the program $P_2$ just emulates the work of $P''$ which is constructed from $P$ by permutation of the transition function of $P$ with respect to the order
            \[
                 \pi = \left(
                          \bin{\bigoplus\limits_{i = 1}^1 \sigma(z_{i, 1}), 
                               \dots,
                               \bigoplus\limits_{i = 1}^1 \sigma(z_{i, l})} + 1,
                      \right.\]\[
                          \dots,\]\[
                      \left.
                          \bin{\bigoplus\limits_{i = 1}^q \sigma(z_{i, 1}), 
                               \dots,
                               \bigoplus\limits_{i = 1}^q \sigma(z_{i, l})} + 1
                      \right).
            \]
            By the definition of commutative $\OBDD[k]$ the program $P''$ computes the same function as $P$. Therefore, $P_2$ returns the same result. And by the definition of the functions $\xorreorder{f}$, $P_1$ computes $\xorreorder{f}$.
            
    \end{description}

    All other cases have the same proofs.
\end{proof}

\begin{corollary}\label{cr:total-reordering}
    Let $f : \bool^q \to \bool$ be a Boolean function, and let $k$ be a positive integer such that
    \begin{itemize}
        \item any $\OBDD[k]$ representation of $f$ has width at least $d$ and
        \item there is a commutative $\OBDD[k]$ ($\NOBDD[k]$) representation of $f$ of width $w$.
    \end{itemize}
    
    Then there is a total Boolean function $g : \bool^n \to \bool$ ($n = q (\ceil{\log q} + 1)$), such that
    \begin{itemize}
        \item $g$ is an extension of the partial function $\reorder{f}$,
         \item any $\OBDD[k]$ representation of $g$ has width at least $d$ and
        \item there is a $\OBDD[k]$ ($\NOBDD[k]$) representation of $g$ of the width $w \cdot q$.       
    \end{itemize}
\end{corollary}
\begin{proof}
    By Theorem~\ref{th:d-obdd}, any $\OBDD[k]$ representation of $\reorder{f}$ has width at least $d$. Due to Theorem \ref{th:deord-classic},  there is a $\OBDD[k]$ representation $P$ of $\reorder{f}$ of width $w \cdot q$. Let $g$ be a total Boolean function such that $g(\sigma) = \reorder{f}(\sigma)$ if $\reorder{f}$ is defined on $\sigma$, otherwise let us define $g(\sigma)$ as $P(\sigma)$.

    Let us note that any $\OBDD[k]$ representation of $g$ also represents $\reorder{f}$; as a result, has width at least $d$. Additionally, let us note that $P$ represents $g$.
\end{proof}

\begin{theorem}\label{th:d-qobdd}
    If there is a bounded error commutative $\QOBDD$ representation of a Boolean function $f : \bool^q \to \bool$ of width $w$, then there is a bounded error $\QOBDD$ representation of a partial Boolean function $\xorreorder{f}$ of width $w \cdot q$.
\end{theorem}
\begin{proof}
    Note that if there is a bounded error commutative $\QOBDD$ representation of $f$ of width $w$, then there is a bounded error $\QOBDD[\pi]$ representation $P$ of $f$ of the same width. For the description of a computation in $P$ we use a quantum register $\ket{\psi} = \ket{\psi_1 \psi_2     \dots \psi_t}$ where $t = \ceil{\log w}$.
    
    Let us consider $\xorreorder{f}$. We construct a bounded error $\QOBDD$ $P'$ for $\xorreorder{f}$ with the following order: $z_{1, 1}, \dots, z_{1, l}$, $y_1$, \dots, $z_{q, 1}$, \dots, $z_{q, l}$, $y_q$. This program uses a quantum register of $\ceil{\log w} + \ceil{\log q}$ qubits, i.e. having $w \cdot q$ states. Let us denote this register as $\ket{\phi} = 
    \ket{\phi_1 \phi_2 \dots \phi_l \psi_1 \psi_2 \dots \psi_t}$. 
    
    The part of the register $\ket\phi$ consisting of  $\ket{\psi_1 \psi_2 \dots \psi_t}$ qubits (we call it as a \textit{computing} part) is modified when $P'$ reads a value bit. Additional qubits $\ket{\phi_1 \phi_2 \ldots \phi_l}$ (we call this part an \textit{address} part) is used to determine address of the value bit. 
    
    Program $P'$ consists of $q$ stages, $i$-th stage corresponds to its own block $z_{i, 1}, \dots, z_{i, l}, y_i$. Informally, when $P'$ processes the block, it stores address in the address part by applying the XOR function to address of the current block. After that, the program modifies the computation part, with respect to the value bit.

    Let us describe $i$-th stage, for $i \in \set{1, \dots, q}$. In the first $\ceil{\log q}$ levels of the stage $P$ computes address $\bin{Adr_i} = \bin{\bigoplus\limits_{j=1}^i z_{i, 1}, \dots, \bigoplus\limits_{j=1}^i z_{i, l}}$. The program reads bits one by one, and for a bit $z_{i, j}$ it applies a unitary operator $U^{z_{i, j}}_j$ on the address part of the register $\ket{\phi}$. Here $U^{z_{i, j}}_j = I \otimes I \otimes \ldots \otimes I \otimes 
        A^{z_{i,j}} \otimes I \otimes \ldots \otimes I $, 
    $A^0 = I$, $A^1 = \mathrm{NOT}$. The matrices $I$ and $\mathrm{NOT}$ are $2 \times 2$ matrices such that $I$ is a diagonal $1$-matrix and $\mathrm{NOT}$ is an anti-diagonal $1$-matrix. We do not modify the computation part on these levels.
    
    Note that after all these operations the address part of the register is equal to $Adr_i$. On the last level we read $y_i$ and transform the register $\ket{\phi}$ by an unitary $(w \cdot q \times w \cdot q)$-matrix $D^{y_i}$ defined in the following way:
    $$
        D^0 = \begin{pmatrix}
                  G_1^0 & 0 & \cdots & 0 \\
                  0 & G_2^0 & \cdots & 0 \\         
                  \vdots & \vdots & \ddots & \vdots \\
                  0 & 0 & \cdots & G_q^0
              \end{pmatrix} \text{ and } 
        D^1 = \begin{pmatrix}
                  G_1^1 & 0 & \cdots & 0 \\
                  0 & G_2^1 & \cdots & 0 \\         
                  \vdots & \vdots & \ddots & \vdots \\
                  0 & 0 & \cdots & G_q^1
               \end{pmatrix},
     $$          
     where $\set{(G_i^0, G_i^1)}_{i = 1}^q$ are unitary matrices for transformation of a quantum system in $P$.

     It is easy to see, that width of $P'$ equals $w \cdot q$. Let us prove that $P'$ represents $\xorreorder{f}$ with bounded error. Let us consider an input $\sigma \in \bool^n$ and let 
     $
         \pi = \left(
                  \bin{Adr_1} + 1,             
                                \dots,
                  \bin{Adr_q} + 1
              \right)
     $
     be an order over the value variables induced by $\sigma$. Since $P$ is a commutative bounded error $\QOBDD$ representation of $f$, we can reorder unitary operators $\set{(G_i^0, G_i^1)}_{i = 1}^q$ according to the order $\pi$ and get a bounded error $\QOBDD[\pi]$ $P_{\pi}$ representation of $f$ as well. It is easy to see that $P'$ emulates exactly the computation of $P_{\pi}$. Therefore $P'$ on $\sigma$ gives us the same result as $P_{\pi}$ on corresponding value bits. Hence, by the definition of $\xorreorder{f}$ we prove that $P'$ represents $\xorreorder{f}$ with bounded error.
\end{proof}

\begin{corollary}
    For any positive $k$, if there is a commutative bounded error $\QOBDD[k]$ of width $d$ representing a Boolean function $f : \bool^q \to \bool$, then there is a bounded error $\QOBDD[k]$ of width $d \cdot q$ representing a partial Boolean function $\xorreorder{f}$.
\end{corollary}

The proof of this corollary is the same as the proof of Theorem~\ref{th:d-qobdd}.

       \section{Commutative $\OBDD$s}\label{sec:com-obdd}
   In this section we discuss a criterion of existence of a small commutative
$\OBDD$ (bounded error $\QOBDD$). We say that a function $f : \bool^n \to \bool$
has a $S_{w, q, \odot}$ representation if there is a sequence of integers
$\set{C_i}_{i = 1}^n$, such that
\[
    f(x_1, \dots, x_n) = 
          q \left(\bigodot\limits_{i = 1}^n C_i x_i \mod w\right),
\]
where $\odot$ is some commutative operation over the set $\set{0, \dots, w - 1}$
and $q : \set{0, \dots, w - 1} \to \bool$.

Let us show that if a function $f$ has a $S_{w, q, \odot}(X)$ representation, then
there is a commutative $\OBDD$ of width $w$ representing $f$.

\begin{theorem}\label{theorem:determenistic-commutative}
    Let $f$ be a Boolean function, such that $f$ has a $S_{w, q, \odot}$
    representation for some $w$, $q$, and $\odot$. Then there is a commutative
    $\OBDD$ of width $w$ representing $f$.
\end{theorem} 
\begin{proof}
    Let us construct such an $\OBDD$ with an order $x_1$, \dots, $x_n$.
    We create a node on level $j$ for each possible value of 
    $\bigodot_{i = 1}^{j - 1} C_i x_i \mod{w}$. Then for each node
    corresponding to $z \in \set{0, \dots, w - 1}$ from $j$-th layer there are
    $1$-edge leads to $z \odot C_j  \mod{w}$ and $0$-edge leads to $z \odot 0$.
    We use $q$ as a function that marks accepting nodes on the last layer.

    By the definition of $S_{w, q, \odot}$ and $\OBDD$ this program
    represents $f(x_1, \dots, x_n)$. Due to the commutativity of $\odot$,
    this $\OBDD$ is commutative.
\end{proof}

Note that any characteristic polynomial, discussed
in~\cite{av2013}, has a $S_{w, q, \odot}$ representation for appropriate $w$,
$\odot$, and $q$.

Let us present the definition of these polynomials. We call a polynomial
$G(x_1, \dots, x_n)$ over the ring $\Z_w$ a characteristic polynomial of a
Boolean function $f(x_1, \dots, x_n)$ if for all  $\sigma \in \{0,1\}^n$,
$G(\sigma) = 0$ holds iff $f(\sigma) = 1$.

Ablayev and Vasilev~\cite{av2013} proved, using the fingerprint technique, the
following result.
\begin{lemma}[\cite{av2013}]\label{lm:av-qf}
    If a Boolean function $f$ has a linear characteristic polynomial over
    $\Z_w$, then the function can be represented by a bounded error quantum
    $\OBDD$ of width $O(\log w)$.
\end{lemma}

It is easy to see that by a linear characteristic polynomial we can construct
$S_{w, q, +}(X)$ representation, where $q$ converts $0$ to $1$ and other values
to $0$. Let us denote such a function $q$ as $q_0$. 

Note that in contrast with Theorem~\ref{theorem:determenistic-commutative},
the quantum fingerprint technique gives us a commutative $\QOBDD$ of a logarithmic
width. Unifying these techniques we can prove the following theorem.

\begin{theorem}\label{theorem:commutative}
    If a Boolean function $f$ has a $S_{w, q_0, +}$ representation
    for some $w$, then there is a commutative bounded error $\QOBDD$
    for $f$ of width $O(\log w)$.
\end{theorem}
\begin{proof}
    Let a Boolean function $f$ has a $S_{w, q_0, +}$ representation for some $w$.
    It means that the function has a linear characteristic polynomial over $\Z_w$. Then by
    Lemma~\ref{lm:av-qf} one may construct a bounded error quantum $\OBDD$ of
    width $O(\log w)$ representing $f$.
\end{proof} 

    \section{Exponential Gap between Quantum and Classical $\OBDD$s}\label{sec:gap}
 As we discussed in the introduction, it is known that the maximal gap between
quantum and deterministic $\OBDD$ complexities of Boolean functions is exponential.

\begin{lemma}[\cite{AGKMP05}]\label{lm:max-gap}
If the best $\OBDD$ representation of a Boolean function
    $f$ has width $2^w$. Then any bounded error $\QOBDD$ for $f$ has width at least $w$.
\end{lemma}

But all the examples that achieve an exponential gap have a sublinear width of a
bounded error quantum $\OBDD$ representation. Known examples with a bigger width
do not achieve this gap. We present results for two functions, based on equality
function that achieve almost exponential gap.

\subsection{Application of Reordering Method} 
Let us apply the reordering method to \textit{equality function}  ($\EQ_n :
\bool^{2n} \to \bool$) where $\EQ_n(x_1, \dots, x_n, y_1, \dots, y_n) = 1$, iff
$x_1 = y_1$, \dots, $x_n = y_n$.

Authors of the paper~\cite{akv2008} have proven that there is a commutative $\QOBDD$ of
width $O(n)$ representing $\EQ_n$ with bounded error. Hence, due to Theorem~\ref{th:d-qobdd}, there is a bounded error $\QOBDD$ of width
$O(q^2)$ for 
$\xorreorder{\EQ_q}$.

It is well-known that any $\OBDD[\id]$ representation of $\EQ_n$ has width at
least $2^n$. As a result, by Theorem~\ref{th:d-obdd} any $\OBDD$ representation
of $\xorreorder{\EQ_q}$ has width at least $2^q$. So, if we apply the Theorems \ref{th:d-obdd} and \ref{th:deord-classic} to the function then we get the following result:

\begin{theorem}
    There is a bounded error quantum $\OBDD$ representation of a partial
    Boolean function $\xorreorder{\EQ_q} : \bool^n \to \bool$ of width
    $O\left(\frac{n^2}{\log^2 n}\right)$, for $n=q\lceil \log q +1 \rceil$; 
    any deterministic $\OBDD$ representation of $\xorreorder{\EQ_q}$ has
    width at least $2^{\Omega\left(\frac{n}{\log n}\right)}$.
\end{theorem}

Let us define \textit{xor-reordered equality function} ($\REQ_q : \bool^n \to
\bool$ where $n = 2q (\ceil{\log 2q} + 1)$). This is a total version of
$\xorreorder{\EQ_q}$. Let us consider
\[
    u(z_{1, 1}, \dots, z_{2q, l}, y_1, \dots, y_{2q}) = 
    \sum\limits_{
        i : Adr(i,z_{1, 1}, \dots, z_{2q, l}) < q
    }
   2^{Adr(i,z_{1, 1}, \dots, z_{2q, l})} y_i \mod 2^q
\] 
and 
\[
    v(z_{1, 1}, \dots, z_{2q, l}, y_1, \dots, y_{2q}) = 
    \sum\limits_{
        i : Adr(i,z_{1, 1}, \dots, z_{2q, l}) \geq q
    }
    2^{Adr(i,z_{1, 1}, \dots, z_{2q, l}) - q} y_i \mod 2^q,
\] 
 where $Adr(i,z_{1, 1}, \dots, z_{2q, l})=\bin{\bigoplus\limits_{j=1}^i z_{j, 1}, \dots, \bigoplus\limits_{j=1}^i z_{j, l}}$
We define $\REQ_q(z_{1, 1}, \dots, z_{2q, l}, y_1, \dots, y_{2q}) = 1$ iff 
$u(z_{1, 1}, \dots, z_{2q, l}, y_1, \dots, y_{2q}) = 
v(z_{1, 1}, \dots, z_{2q, l}, y_1, \dots, y_{2q})$. The function has two following properties:
\begin{lemma}\label{lm:deq-dlower}
    Any $\OBDD$ representation of $\REQ_q$ has width at least
    $2^{\Omega\left(\frac{n}{\log n }\right)}$, where $n=2q\lceil 1+\log 2q \rceil$ is the length of an input.
\end{lemma}
\begin{proof}
    Note that $\REQ_q$ is an extension of $\xorreorder{\EQ_q}$. Thus any
    $\OBDD$ representation of $\REQ_q$ also represents $\xorreorder{\EQ_q}$.
    Hence, by Theorem~\ref{th:d-obdd} any $\OBDD$ representation of $\REQ_q$ has
    width at least $2^q \ge 2^{\Omega\left(\frac{n}{\log n }\right)}$.
\end{proof}

\begin{theorem} \label{th:deq}
    There is a bounded error quantum $\OBDD$ representation of $\REQ_q$ of
    width $O\left(\frac{n^2}{(\log n)^2}\right)$, where $n = 2q\lceil 1+\log 2q \rceil$ is the length of an input.
\end{theorem}
\begin{proof}
    A computing
    of the Boolean function $\REQ_q$ is equivalent to checking the equality of
    $g(z_{1, 1}, \dots, z_{2q, l}, y_1, \dots, y_{2q}) = c_1(z_{1, 1}, \dots, z_{2q, l})\cdot y_1 + \dots + c_{2q}(z_{1, 1}, \dots, z_{2q, l})\cdot y_{2q}$ and $0$. So we
    can construct a commutative $\QOBDD$ $P$ with one side error using the quantum
    fingerprinting technique, similar to the proof of Lemma \ref{lm:av-qf}.

    By the definition, $\REQ_q(z_{1, 1}, \dots, z_{2q, l}, y_1, \dots, y_{2q}) = 1$
    iff
    \[
        \sum\limits_{
            i : Adr(i,z_{1, 1}, \dots, z_{2q, l})< q
        } 2^{Adr(i,z_{1, 1}, \dots, z_{2q, l})} y_i
        - \\
        \sum\limits_{
                i : Adr(i,z_{1, 1}, \dots, z_{2q, l}) \geq q
            } 2^{Adr(i,z_{1, 1}, \dots, z_{2q, l}) - q} y_i 
        \equiv 
        0 
        \pmod{2^q}.
    \]

    Let us take QOBDD $P'$ for $EQ(x_1,\dots,x_{2q})$ from Theorem 2 of \cite{akv2008}. This program is based on fingerprinting technique and checking $g'(x)=0 \pmod{2^{q}}$ for $g'(x)=\sum\limits_{
            i=1
        }^{q} 2^{i-1} x_i
        -
        \sum\limits_{
                i =q+1
            }^{2q} 2^{i-q-1} x_i$. This program is commutative and use $\log_2 t+1$ qubits, where $t=O(q)$. The $j$-th angle that is used for rotating about the $\hat{y}$ axis of the Bloch sphere on $x_i$  is $\alpha_{i,j}$, for $j\in\{1,\dots,t\}$. Here $\alpha_{i,j}=\left(4\pi k_j 2^{i}\right)/2^{q}$, if $i< q$; $\alpha_{i,j}=-\left(4\pi k_j 2^{i-q}\right)/2^{q}$, otherwise.
            In the end of the rotation process for the program is state
            \[|\psi\rangle = \frac{1}{\sqrt{t}}\sum_{j=1}^t|i\rangle \left(cos\frac{2\pi k_j g'(x)}{2^{n/2}}|0\rangle+ sin\frac{2\pi k_j g'(x)}{2^{n/2}}|0\rangle\right)\]
            
            So, we apply xor-reordering method to this commutative QOBDD and we use different angels for different addresses. Therefore the angles becomes $\alpha'_{i,j}$, for $j\in\{1,\dots,t\}$. Here $\alpha'_{i,j}=\left(4\pi k_j 2^{Adr(i,z_{1, 1}, \dots, z_{2q, l})}\right)/2^{q}$, if $i< q$; $\alpha'_{i,j}=-\left(4\pi k_j 2^{Adr(i,z_{1, 1}, \dots, z_{2q, l}) -q}\right)/2^{q}$, otherwise. Hence we get the following state in the end of the rotation process of the quantum QOBDD $P$:
            \[|\psi\rangle = \frac{1}{\sqrt{t}}\sum_{j=1}^t|i\rangle \left(cos\frac{2\pi k_j g(x)}{2^{n/2}}|0\rangle+ sin\frac{2\pi k_j g(x)}{2^{n/2}}|0\rangle\right)\]
            
            After that $P$ applies remaining steps of fingerprinting techniques and gets the answer with bounded error.
            Let us compute width of $P$. Width of $P'$ is $O(q)$. Because of reordering method, width of $P$ is $O(q^2)\leq O(n^2/(\log n)^2)$.
\end{proof}

\subsection{Shifted Equality}

In order to get another separation between quantum and classical $\OBDD$
complexities let us consider \textit{shifted equality function} ($\SEQ_q :
\bool^{2q + l} \to \bool$ where $l = \ceil{\log q}$), the function introduced by
JaJa, Prasanna, and Simon~\cite{JKS84}. The function is defined in the following
manner:
$\SEQ_n(x_1, \dots, x_q, y_1, \dots, y_q, s_1, \dots, s_l) = 1$ 
iff for all $i \in \{1,\dots,q\}$, 
$
    x_i = y_{(i + \bin{s_1, \dots, s_l}) \pmod{q}}.
$

Using a lower bound for the best communication complexity of this
function~\cite{JKS84} and the well-known connection between $\OBDD$ and
communication complexities we have the following property.
\begin{lemma}[see for example~\cite{KN97}]
        Any $\OBDD$ representation of $\SEQ_q$ has a size at least $2^{\Omega(n)}$, where $n=2q+\ceil{\log q}$ is the length of an input.
\end{lemma}

We can also construct a bounded error quantum $\OBDD$ representation of $\SEQ_q$ of a
small width.
\begin{lemma}
     There is a bounded error quantum $\OBDD$ representation for $\SEQ_q$ of
     width $O(n^2)$, where $n=2q+\ceil{\log q}$ is the length of an input.
\end{lemma}
\begin{proof}
    Let us construct a $\QOBDD$ $P$ that reads an input in the following order: 
    $s$, then $x$, and then $y$; also $P$ uses a quantum register consisting of
    two parts: the first part $\ket{\phi}$ is for storing the value of the shift
    ($\bin{s_1, \dots, s_l}$) and the second one $\ket{\psi}$ is called
    a computational part. The size of $\ket{\phi}$ is $\ceil{\log q}$ qubits and
    the size of $\ket{\psi}$ is $\log q + C$, for some constant $C$.

    On the first $\ceil{\log q}$ levels, the program stores input bits
    into $\ket{\phi}$ using a storing procedure similar to procedure from the
    proof of Theorem~\ref{th:d-qobdd}.

    Then we apply the fingerprint algorithm from~\cite{AV08,av2009}, but use
    unitary matrices for $y$ with shift depending on the state of $\ket{\phi}$.

    After reading the last variable, we measure $\ket{\psi}$ and get the
    answer.

    The width of the program is $2^{\ceil{\log q} + \log q + C} = O(q^2)= O(n^2)$.
\end{proof}

It is interesting to compare this separation and the separation obtained in the
previous subsection. In the case of $\SEQ_q$, the lower bound for $\OBDD$ width is
$2^{\Omega(n)}$, but for $REQ$ function, it is $2^{\Omega(\frac{n}{\log n})}$.
On the other hand, the upper bound for the width of $\QOBDD$ is also larger.

    \section{Hierarchy for Probabilistic and Quantum
        $\OBDD$s}\label{sec:hierarchies}
    In this section, we consider classes $\BPOBDD_d$ and $\BQOBDD_d$ of Boolean functions that can be represented by bounded error probabilistic and quantum $\OBDD$s of width $O(d)$, respectively. We prove hierarchies for these classes with respect to
$d$.

\subsection{Hierarchy for Probabilistic $\OBDD$s.} 
Before we start proving the hierarchy let us consider a Boolean function
$\WS_n$, or \textit{weighted sum} function introduced by Savick{\`y} and
{\v{Z}}{\'a}k~\cite{sz2000}.

Let $n > 0$ be an integer and let $p(n)$ be the smallest prime greater than $n$.
Let us define functions $s_n : \bool^n \to \bool$ and $\WS_n : \bool^n \to
\bool$, such that
$s_n(x_1, \dots, x_n) = \left(\sum\limits_{i=1}^{n} i \cdot x_i\right) \mod p(n)$ 
and $\WS_n(x_1, \dots, x_n) = x_{s_n(x_1, \dots, x_n)}$. For the function
$\WS_n$ it is known \cite{sz2000} that any bounded error probabilistic $\OBDD$ representing
$\WS_n$ has width at least $2^{\Omega(n)}$.

Let $\set{f_n : \bool^n \to \bool}_{n \in \N}$ be a family of Boolean functions
and $b : \N \to \N$ be a function such that $b(n) \le n$. 
We denote by $\set{f^b_n : \bool^n \to \bool}_{n \in \N}$
the family of Boolean functions such that $f^b_n(x_1, \dots, x_n) =
f_{b(n)}(x_1, \dots, x_{b(n)})$. 
\begin{remark}
\label{remark:padding}
    If for any $\OBDD$ (bounded error $\POBDD$ or $\QOBDD$) representation of $f_n$
    has width at least $w(n)$, then $\OBDD$ (bounded error $\POBDD$ or $\QOBDD$)
    representation of $f^b_n$ has width at least $w(b(n))$. Moreover, if there is an
    $\OBDD$ (bounded error $\POBDD$ or $\QOBDD$) representation of $f_n$ of width
    $d(n)$, then there is an $\OBDD$ (bounded error $\POBDD$ or $\QOBDD$)
    representation of $f^b_n$ of width $d(b(n))$.
\end{remark}

We apply this construction and get modification of the Boolean function $\WS_n$ using padding bits. We will denote
this modified function as $\WS^b_n$. Let $n > 0$ and $b > 0$ be integers, such
that $b \le \frac{n}{3}$ and $p(b)$ be the smallest prime greater than $b$.
We denote by $\WS^b_n : \bool^n \to \bool$ a function such that 
$\WS^b_n(x_1, \dots, x_n) = x_{s_b(x_1, \dots, x_b)}$.

In order to use Remark \ref{remark:padding} we need the following lemma:
\begin{lemma}[\cite{sz2000}]\label{lm:ws-base}
    Any bounded error probabilistic
    $\OBDD$ that represents $\WS_n$ has width at least $2^{\Omega(n)}$ and
    there is a bounded error probabilistic $\OBDD$ of width $2^n$ representing
    $\WS_n$.
\end{lemma}

Let us prove the hierarchy theorem for $\BPOBDD_d$ classes using these
properties of the Boolean function $\WS_n^b$.
\begin{theorem}\label{th:prob-hi}
    If $d$ and $\delta$ are functions such that $d(n) = o(2^n)$,
    $d(n) = \omega(1)$, and $\delta(n) = \omega(1)$, then
    $\BPOBDD_{d^{1 / \delta}} \subsetneq \BPOBDD_d$.
\end{theorem}
\begin{proof}
    It is easy to see that $\BPOBDD_{d^{1 / \delta}} \subseteq \BPOBDD_d$. Let
    us prove the inequality of these classes. Due to Lemma~\ref{lm:ws-base} and Remark \ref{remark:padding}, the
    Boolean function $\WS_n^{\log d} \in \BPOBDD_d$. However, any bounded error
    probabilistic $\OBDD$ representing $\WS_n^{\log d}$ has width
    $2^{\Omega(\log d)}$ that is greater than $d^{1 / \delta}$ since $d =
    \omega(1)$. Therefore, $\WS_n^{\log d} \not\in \BPOBDD_{d^{1 / \delta}}$.
\end{proof}

\subsection{Hierarchy for Quantum $\OBDD$s.}
In this subsection we consider similar modifications of three well-known functions:
$\REQ_n$, $\MOD_{p, n}$, and $\MSW_n$ (defined in~\cite{s06}).
The function $\MSW_n$ may be defined in the following way:
$\MSW_n (x_1, \dots, x_n) = x_z \oplus x_{r + n / 2}$, where 
$z = s_{n / 2}(x_1, \dots, x_{n / 2})$ and 
$r = s_{n / 2}(x_{n / 2 + 1}, \dots, x_n)$, if  $r = z$ and 
$\MSW_n(x_1, \dots, x_n) = 0$ otherwise.

We will use Remark \ref{remark:padding}. So, we need the following two lemmas.
\begin{lemma}[\cite{s06}]
\label{lemma:msw}
    Any bounded error quantum $\OBDD$ representation of $\MSW_n$ has width
    at least $2^{\Omega(n)}$ and there is a bounded error quantum $\OBDD$ of
    width $2^n$ representing $\MSW_n$.
\end{lemma}

\begin{lemma}[\cite{AGKMP05,AV08}]
\label{lemma:mod}
    Any bounded error quantum $\OBDD$ representation of $\MOD_{p, n}$ 
    (for $p \le n$) has width at least $\floor{\log p}$ and there is a bounded
    error quantum $\OBDD$ of width $O(\log p)$ representing $\MOD_{p, n}$.
\end{lemma}

Now we are ready to prove the main theorem of this section.
\begin{theorem}\label{th:quntum-hi}
    Let $d : \N \to \N$ and $\delta : \N \to \N$ be functions such that $d(n) =
    \omega(1)$ and $\delta(n) = \omega(1)$.
    \begin{itemize}
        \item If $d(n) \le \log n$ for all $n$, then
            $\BQOBDD_{\frac{d}{\delta}} \subsetneq \BQOBDD_{d}$;
        \item If $d(n) \le n$ for all $n$, then
            $\BQOBDD_{\frac{d}{\log^2 d}} \subsetneq \BQOBDD_{d^2}$;

        \item If $d(n) \le 2^n$ for all $n$, then
            $\BQOBDD_{d^{1/\delta}}\subsetneq \BQOBDD_d$.
    \end{itemize}
\end{theorem}
\begin{proof}
    It is easy to see that for any $d' \le d$, $\BQOBDD_{d'} \subseteq
    \BQOBDD_{d}$. Let us prove the inequalities.

    Due to Lemma~\ref{lemma:mod}, the Boolean function $\MOD_{2^d, n} \in
    \BQOBDD_d$. However, width of any bounded error quantum $\OBDD$ representing
    $\MOD_{2^d, n}$ is $\Omega(d)$. Therefore $\MOD_{2^d, n} \not\in
    \BQOBDD_{d / \delta}$.

    Due to Theorem~\ref{th:deq}, the Boolean
    function $\REQ_n^d \in \BQOBDD_{d^2}$. On the contrary by Lemma \ref{lm:max-gap}, Lemma~\ref{lm:deq-dlower}
    and Remark~\ref{remark:padding} width of any bounded error quantum $\OBDD$
    representing $\REQ_n^d$ is at least $\floor{\frac{d}{\ceil{\log d + 1}}}$.
    Therefore, $\REQ_n^d \not\in \BQOBDD_{\frac{d}{\log^2 d}}$.

    Due to Lemma~\ref{lemma:msw} and Remark \ref{remark:padding}, the Boolean function 
    $\MSW_n^{\log d} \in \BQOBDD_d$. However, width of any bounded error
    quantum $\OBDD$ representing $\MSW_n^{\log d}$ is at least 
    $2^{\Omega(\log d)}$. Therefore,
    $\MSW_n^{\log d} \not\in \BQOBDD_{d^{1 / \delta}}$.
\end{proof} 

    \section{Extension of Hierarchies for Deterministic and Probabilistic
        $\OBDD[k]$s}\label{sec:kobdd-hrch}
 This section shows the separation between $\OBDD[k]$s and $\OBDD[2k]$s
using the reordering method and a lower bound for complexity of
\textit{pointer jumping} function also denoted as $\PJ$~\cite{NW93,bssw96}

At first, let us present a version of the pointer jumping function which works
with integer numbers. Let $V_A$ and $V_B$ be two disjoint sets of vertices with
$|V_A| = |V_B| = m$ and $V = V_A \cup V_B$ . 
Let $F^A = \set{f^A : V_A \to V_B}$, $F^B = \set{f^B : V_B \to V_A}$ and 
$f = (f^A, f^B) : V \to V$ defined by the following rule: 
\begin{itemize}
    \item if $v \in V_A$, then $f(v) = f^A(v)$ and
    \item if $v \in V_B$, then $f(v) = f^B(v)$. 
\end{itemize}
For each $k \ge 0$ we define $f^{(k)}(v)$ such that $f^{(0)}(v) = v$ and 
$f^{(k + 1)}(v) = f(f^{(k)}(v))$. Let $v_0 \in V_A$, The function we are
interested in is $g_{k, m} : F^A \times F^B \to V$ such that 
$g_{k, m}(f^A, f^B) = f^{(k)}(v_0)$. 

The Boolean function $\PJ_{k,m} : \bool^{2m\lceil\log_2m\rceil} \to \bool$ is a Boolean version of
$g_{k, m}$ where we encode $f^A$ as a binary string using $m \lceil\log_2m\rceil$ bits and
$f^B$ as well. The result of the function is the parity of bits of the binary
representation for the resulted vertex.

We apply the reordering method to the $\PJ_{k, m}$ function and call the total
version of it, obtained from Corollary~\ref{cr:total-reordering}, as 
$\RPJ_{k, m}$.

Note that to prove an upper bound for $\RPJ_{2k - 1, m}$ it is necessary
to construct a commutative $\OBDD[2k]$ for $\PJ_{2k - 1, m}$. In order to
prove a lower bound for $\RPJ_{2k - 1, m}$ it is necessary to prove
a lower bound for $\PJ_{2k - 1, m}$.

For proving the lower bound, we need a notion of communication complexity.
Let $h  : \bool^q \times \bool^m \to \bool$ be a Boolean function. We have two
players called Alice and Bob, who have to compute $h(x, y)$. The function $h$ is
known by both of them. However, Alice knows only bits of $x$ and Bob knows only
bits of $y$. They have a two-sided communication channel. On each round of their communication, one of them sends a string. Alice and Bob are trying to minimize two parameters: a total number of sent bits and number of rounds.  For the formal
definition see for example~\cite{KN97}.

Additionally, we say that the $k$-round communication complexity with Bob sending
first of a function $h(x_1, \dots, x_q, y_1, \dots, y_m)$ equals to $c$ iff the minimal number of sent bits of $k$-round communication protocols with Bob sending first is equal to $c$. We denote this complexity as $\C[B, k]{h}$ if this protocol is deterministic and $\C[B, k]{h}[\epsilon]$ for probabilistic one with bounded error $\epsilon$.

\begin{lemma}[{\cite{NW93}}]\label{lemma:communication-hierarchy-determenistic}
    $\C[B, k]{\PJ_{k, m}} = \Omega(m - k \log{m})$ for any $k$.
\end{lemma}

\begin{lemma}[{\cite{NW93}}]\label{lemma:communication-hierarchy-prob}
    $\C[B, k]{\PJ_{k, m}}[1/3] = \Omega(\frac{m}{k^2} - k \log{m})$ for any $k$.
\end{lemma}

Note that there is a well-know connection between communication complexity and
$\OBDD$ complexity.
\begin{lemma}[see for example \cite{Kha16}]
    Let $h(x_1, \dots, x_q, y_1, \dots, y_m)$ be a Boolean function, $\pi$ be
    an order over the variables $x_1, \dots, x_q, y_1, \dots, y_m$ such that
    $y_i$ precedes $x_j$ for any $i \in \set{1, \dots, q}$, and 
    $j \in \set{1, \dots, m}$.

    If there is a $\OBDD[k][\pi]$ representing $h$ of width $2^w$, then there is a
    $(2k - 1)$-round communication protocol for $h$ of cost $w$  such that Alice has  $x_1, \dots, x_q$, Bob has $y_1, \dots, y_m$ and sends
    a first message.
\end{lemma}

The next corollary follows from the previous three lemmas. 
\begin{corollary}\label{cr:pj-lower}
    For any positive integer $k$ and order $\pi$, such that the variables encoding
    $f_A$ precedes the variables encoding $f_B$,
    \begin{itemize}
        \item width of any $\POBDD[k][\pi]$ representing $\PJ_{2k - 1, m}$ with
            bounded error is at least $2^{\Omega(\frac{m}{k^2} - k
            \log m)}$ and 
        \item width of any $\OBDD[k][\pi]$ representing $\PJ_{2k - 1, m}$ is at
            least $2^{\Omega(m - k \log{m})}$.
    \end{itemize}
\end{corollary}

\begin{lemma}\label{lm:dpj}
    There is a commutative $\OBDD[2k]$ representing $\PJ_{2k - 1, m}$ of width
    $O(n^2)$, where $n=2m\ceil{\log m}$ is a length of an input.
\end{lemma}
\begin{proof} 
    First of all, let us consider two stets $T_1,T_2\in\{0,\dots,w-1\}$ for some integer $w$. Let us note that if a function $h : \bool^n \to T_1$ has a commutative
    $\OBDD[k']$ and for all $t \in T_1$, a function $h'_t : \bool^n \to T_2$ has a
    commutative $\OBDD$. Then a function $h'' : \bool^n \to T_2$ such that $h''(x) =
    h'_{h(x)}(x)$ has a commutative $\OBDD[(k' + 1)]$.

    Secondly, let us describe a commutative OBDD $P$ for computing $f(v)$. The program has $m$ different branches for each value of $v$. The $z$-th brunch for $v=z$ has $m$ nodes on each level. Each node corresponds to each of $m$ possible results of $f$. Initially the brunch is in node that corresponds to $0$-value of the function. The brunch skips all variables except variables $x^z_1,\dots x^z_{\lceil\log_2 m\rceil}$ that represent value $f(z)$. This brunch computes $u=\sum\limits_{j=1}^{\lceil\log_2 m\rceil}2^{j-1}x^z_j$ $(mod$ $m)$ and reaches a node that corresponds to $u$. The result is $f(z)=u$. Note that skipping is adding $0$ to $u$.
    
In the initial step, $P$ starts on one of the brunches that corresponds to $v$. Each bunch has width $O(m)$ and $P$ has width $O(m^2)=O(n^2)$.

It is easy to see that each brunch is commutative OBDD and whole quantum OBDD $P$ is also commutative because the sum is commutative operation. Therefore, a permutation of transition functions does not affect the result.    
\end{proof}

\begin{corollary}\label{cr:dpj-upper}
    There is a $\OBDD[2k]$ representing $\RPJ_{2k - 1, n}$ of width $O(n^3)$, where $n=2m\ceil{\log m}$ is a length of an input.
\end{corollary}
\begin{proof}
The claim follows from Lemma \ref{lm:dpj} and Theorem \ref{th:deord-classic}.
\end{proof}
Using this results we can extend the hierarchy for following classes:
$\P{\OBDD[k]}$, $\BPP{\beta}{\OBDD[k]}$, $\SUPERPOLY{\OBDD}$,
$\BSUPERPOLY{\beta}{\OBDD[k]}$, $\SUBEXP{\alpha}{\OBDD[k]}$, and
$\BSUBEXP{\alpha}{\beta}{\OBDD[k]}$. These are classes of Boolean functions
computed by the following models:
\begin{itemize}
    \item $\P{\OBDD[k]}$ and  $\BPP{\beta}{\OBDD[k]}$ are for polynomial width
        $\OBDD[k]$, the first one is for deterministic case and the second one
        is for bounded error probabilistic $\OBDD[k]$ with error at most
        $\beta$.
    \item $\SUPERPOLY{\OBDD[k]}$ and $\BSUPERPOLY{\beta}{\OBDD[k]}$ are similar
        classes for superpolynomial width models.
    \item $\SUBEXP{\alpha}{\OBDD[k]}$ and $\BSUBEXP{\alpha}{\beta}{\OBDD[k]}$ are
        similar classes for width at most $2^{O(n^\alpha)}$, for $0 < \alpha <
        1$.
\end{itemize}

Note that, a size $s$ of an $\OBDD[k]$ is such that $w\leq s\leq n\cdot w$, where $w$ is width of the  $\OBDD[k]$ and $n$ is length of an input. That is why, if an $\OBDD[k]$ has polynomial width, then it has polynomial size. The same situation with superpolynomial and subexponential width and size.
\begin{theorem}\label{thm:hierarchy}
    \begin{enumerate}
        \item $\P{\OBDD[k]} \subsetneq \P{\OBDD[2k]}$, for $k = o(n/\log^3 n)$.
        \item $\BPP{1/3}{\OBDD[k]} \subsetneq \BPP{1/3}{\OBDD[2k]}$, for 
            $k = o(n^{1/3}/\log n)$.
        \item $\SUPERPOLY{\OBDD[k]} \subsetneq \SUPERPOLY{\OBDD[2k]}$, for 
            $k = o(n^{1-\delta})$, $\delta > 0$. 
        \item $\BSUPERPOLY{1/3}{\OBDD[k]} \subsetneq 
                   \BSUPERPOLY{1/3}{\OBDD[2k]}$, for 
              $k = o(n^{1 / 3 - \delta})$ and $\delta > 0$.
            
          \item $\SUBEXP{\alpha}{\OBDD[k]} \subsetneq \SUBEXP{\alpha}{\OBDD[2k]}$, for 
              $k = o(n^{1 - \delta})$,  $1 > \delta > \alpha + \varepsilon$,
              and $\varepsilon > 0$. 
          \item $ \BSUBEXP{\alpha}{1/3}{\OBDD[k]} \subsetneq
                  \BSUBEXP{\alpha}{1/3}{\OBDD[2k]}$, for 
                $k = o(n^{1 / 3 - \delta / 3})$, $1 / 3 > \delta > \alpha +
                \varepsilon$, and $\varepsilon > 0$.
    \end{enumerate}
\end{theorem}
\begin{proof}
    Proofs of all statements are the same. Here we present only a proof of  the first one.

    Let us consider $\RPJ_{2k-1,n}$. Every $\OBDD[k]$ representing the
    function has width at least
    \[
        2^{\Omega(n/(k\log n) - \log (n / \log n))} \ge \\
        2^{\Omega(n / (n\log^{-3} n \log n)- \log (n/\log n))} = \\
        2^{\Omega(\log^2n)} = n^{\Omega(\log n)},
    \]
    due to Corollary \ref{cr:pj-lower} and Theorem \ref{th:d-obdd}.
    Therefore, it has more than polynomial width. Hence, $\RPJ_{2k - 1, n}
    \not\in \P{\OBDD[k]}$ and $\RPJ_{2k - 1,n} \in \P{\OBDD[2k]}$, due to
    Corollary~\ref{cr:dpj-upper}.
\end{proof}

	\section{Conclusion}\label{sec:conc}
In the paper, we present a new technique (``reordering'') for proving lower bounds and upper bounds for OBDD with an arbitrary order of variables if we have similar bounds for the natural order. Based on this technique we show the almost exponential gap between complexity (width) of quantum and deterministic OBDDs. The separation is better that existing once that $\PERM_n$ and $\MOD_{p, n}$ functions provide. The gap is $2^{\Omega(\frac{n}{\log n})}$ vs $O(\frac{n^2}{\log^2 n})$.
As an alternative function we present  \textit{shifted equality function}  that also demonstrates almost exponential gap between deterministic and quantum complexity. The gap is $2^{\Omega(n)}$ vs $O(n^2)$.

The open problem is to suggest a function that demonstrates a complexity gap $2^{\Omega(n)}$ vs $O(n)$.

Based on the ``reordering'' technique we present a function that allow us to prove a width hierarchy for classes of Boolean functions computed by bounded error quantum $\OBDD$s. It is tight in the case of $o(\log n)$ width and not tight in the case of $2^{O(n)}$ width.

The open problem is to suggest a tight width hierarchy for  $2^{O(n)}$ width.

We show  $\P{\OBDD[k]} \subsetneq \P{\OBDD[2k]}$ hierarchy by a number of tests for polynomial size deterministic $\OBDD[k]$s  for $k = o(n / \log^3 n)$. It improves existing hierarchies from \cite{bssw96,Kha16}. 

The open problem is to prove the tight hierarchy for $k = \omega(\sqrt{n} \log^{3 / 2} n)$.

We show the similar hierarchy for probabilistic oblivious $\BP[k]$ that improves result from \cite{hs2003}.

The open problem is to prove the tight hierarchy for $k = \omega(\log \frac{n}{3})$.

    \section*{Acknowledgements.}
   This paper has been supported by the Kazan Federal University Strategic Academic Leadership Program ("PRIORITY-2030").

We thank Alexander Vasiliev and Aida Gainutdinova from Kazan Federal
University and Andris Ambainis from University of Latvia for their helpful
comments and discussions.

    \bibliographystyle{splncs04}
    \bibliography{main}
\end{document}